\documentclass[aps,showpacs,pra,twocolumn]{revtex4-1}

\usepackage[normalem]{ulem}
\usepackage{exscale}
\usepackage{graphicx}
\usepackage{amsmath}
\usepackage{latexsym}
\usepackage{amsfonts}
\usepackage{amssymb}
\usepackage{times}
\usepackage[T1]{fontenc}

\usepackage{amsthm}

\usepackage{epsfig,pstricks}

\newtheorem{theorem}{Theorem}
\newtheorem{lemma}{Lemma}

\begin{document}

\title{Maximal entanglement of two delocalized spin-$\frac{1}{2}$ particles}

\author{Markus Johansson} 
\affiliation{ICFO-Institut de Ciencies Fotoniques, The Barcelona Institute of Science and Technology,
08860 Castelldefels (Barcelona), Spain}


\date{\today}

\begin{abstract}
We describe the entanglement of two indistinguishable delocalized spin-$\frac{1}{2}$ particles in the simplest spatial configuration of three spatial modes with the constraint that at most one particle occupy each mode. It is show that this is the only number of modes for which maximally entangled states exist in such a system. The set of entangled states, including the set of maximally entangled states, is described and different types of entanglement in terms of Bell-nonlocal correlations for different partitions of the system are identified. In particular we focus on the entangled states that are Bell-local for a tri-partition of the system and cannot be described as a superposition of bi-partite entangled pairs of localized particles. Two entanglement invariants are constructed and it is shown that all entanglement monotones are functions of these.
Furthermore, the system has a generic non-trivial local unitary symmetry with a corresponding $2\pi/3$ fractional topological phase. In addition to this a necessary and sufficient condition for the existence of maximally entangled states in systems of arbitrary numbers of delocalized particles with arbitrary spin, where at most one particle can occupy each mode, is derived.
\end{abstract}

\pacs{03.67.-a, 03.67.Mn, 03.65.Ud, 03.65.Vf}

\maketitle
\section{Introduction}

Entanglement is a feature of quantum mechanics that does not exist in classical mechanics. As such, it is responsible for many phenomena without a classical counterpart.
This has made entanglement a resource for quantum information processing tasks, including teleportation \cite{bennett1} and dense coding \cite{bennett2}. There is also evidence that entanglement plays a role in many-body phenomena such as magnetism \cite{brukner,christensen}. Another effect of entanglement is Bell-nonlocality \cite{Bell}, i.e., correlations between distant events that can not be explained by a local causal model.
Bell-nonlocality in itself has been realized to be a resource for quantum informations tasks, e.g., device independent quantum information protocols \cite{mayersyao,colbeck,acin0}.

The entanglement between a small number of particles with well defined locations has been extensively studied and characterized  \cite{mahler,wootters,popescu2,coffman,dur,carteret,sudbery,acin2,acin,luquethibon1,verstraete2002,
Osterlohsiewert2005,kraus}. The tools and concepts used to characterize entanglement in these systems includes Stochastic Local Operations and Classical Communication (SLOCC) \cite{bennett,dur}, entanglement invariants \cite{linden,grassl,wootters,popescu2,coffman,sudbery,luquethibon1,Osterlohsiewert2005} and measures \cite{vidal}, and maximally entangled states \cite{mahler}.

However, in systems of indistinguishable particles where the spatial locations of the particles are not well defined the entanglement between individual particles does not have an operational meaning. This type of systems have been less studied.  
A natural way to describe entanglement in this case is to consider the modes of the system accessible to observers rather than the particles themselves as the operationally relevant objects which can be entangled \cite{knill,zanardi2,zanardi,fisher,zanardi3,barnumviola}.

If the number of available spatial modes is larger than the number of particles, it is possible for the particles to delocalize over the modes. Systems of this type include particles delocalized in a lattice, such as quasi-particle excitations in a crystal lattice, or atoms in an optical lattice \cite{bloch,plenio,gross}. Other examples of spatial modes where particles can delocalize are the traces of a quantum circuit or the optical paths in a photonic device.

If the spatial modes are sufficiently small compared to the effective size of the particles so that their interaction is non-negligible, a repulsive interaction can prevent more than one particle from occupying the same mode. Such a repulsive interaction can be for example Coulomb repulsion, or in the case of composite particles made up of constituent fermion particles, the Fermi-repulsion between the constituent particles \cite{dyson}.  An example of a simple model with this type of behaviour is the Hubbard model \cite{essler} with a strong repulsive on-site interaction term. 
Entangled states created under this type of conditions would have at most one particle per mode.




The entanglement of a system of delocalized particles has a richer structure than a system of the same number of localized particles since the internal degrees of freedom of a particle in a given mode can be entangled not only with internal degrees of freedom of a different particle but also with its spatial degree of freedom. The concepts and tools used to characterize the entanglement of localized particles are operationally well defined and can be applied also to the case of delocalized particles. However, the properties of such systems are in some ways qualitatively different.

In this paper we consider the entanglement of two indistinguishable delocalized spin-$\frac{1}{2}$ particles, or more generally any particles with two internal degrees of freedom. Moreover, we assume that at most one particle can exist in any given spatial mode.
In particular we study the simplest such system where three spatial modes are available. 
It is shown that for systems of two  such delocalized particles, where at most one particle can occupy each spatial mode, this is the only number of modes for which entanglement monotones, i.e., measures of entanglement, can be defined. Therefore, it is also the only system of this kind where maximally entangled states exist. Furthermore, the delocalization of the entangled particles gives rise to a generic non-trivial intrinsic symmetry. To this symmetry corresponds a fractional topological phase $2\pi/3$.

We discuss the role of the conservation of particle number and identify the group of SLOCC in the present context. 
Following this we construct two invariants under SLOCC from which entanglement monotones can be derived. These two invariants generate the full set of SLOCC invariants.
Furthermore, we parametrize the set of entangled states and
characterize different types of entanglement in terms of their Bell-nonlocal correlations for different partitions of the modes of the system and between different parties. 
In particular we consider the entangled states which cannot be understood as super-positions of pairs of bipartite entangled localized particles. These states are Bell-local for a tri-partition of the system.

We show that a non-zero value of an entanglement monotone implies that the state is Bell-nonlocal over every bi-partition. Therefore, in particular, all maximally entangled states are Bell-nonlocal for every bi-partition. We describe the set of maximally entangled states, including the unique maximally entangled state which is Bell-local for a tri-partition of the system.

In addition to the study of two spin-$\frac{1}{2}$ particles, we consider the general condition for the existence of maximally entangled in a system of an arbitrary number of delocalized particles with at most one particle per mode. A necessary and sufficient criterion for the existence of maximally entangled states is given as a relation between the number of particles, the number of modes, and the spin of the particles. Moreover, it is shown that in this general case a non-zero value of an entanglement monotone implies that the state is Bell-nonlocal over every bi-partition of the system into a single spatial mode and the remaining modes.


The outline of the paper is the following. In section \ref{sec2} we define the system and discuss the role of the particle conservation super-selection rule for entanglement and Bell-nonlocality. We then describe the group of SLOCC and construct the invariants under its action. Section \ref{sec3} contains a description of the set of entangled states, with a focus on the Bell-local states for a tri-partition of the system, and a brief description of the generic intrinsic symmetry of the system. In section \ref{sec4} the set of maximally entangled states is described. In Sect. \ref{gen} the general criterion for the existence of maximally entangled states in systems of delocalized particles with at most one particle per mode is given. The paper ends with the conclusions.

\section{Entanglement and Bell-nonlocality of delocalized particles}
\label{sec2}

The scenario we consider is a system of two indistinguishable particles, each with two internal degrees of freedom. The example used is two spin-$\frac{1}{2}$ particles and we will therefore denote the internal states by $|\!\uparrow\rangle$ and $|\!\downarrow\rangle$, but the results of the paper are independent of the  physical nature of the two internal degrees of freedom.
We assume that the particles have access to a finite number of spatial modes and that their locations are not well defined, i.e., they are delocalized. To each spatial mode we associate a party, i.e., an observer, capable of performing measurements on a particle that occupies its spatial mode. We also occasionally consider parties with access to more than one spatial mode.
Since the particle locations are not well defined it is not operationally meaningful to consider entanglement between degrees of freedom of the individual particles. Instead we consider entanglement between the degrees of freedom associated with the parties \cite{knill,zanardi2,zanardi,fisher,zanardi3,barnumviola}, i.e., entanglement between the modes \cite{zanardi,fisher}. What defines the different parts of the system that can be entangled is thus the subsets of operations that the different parties can perform \cite{knill,zanardi2,zanardi3,barnumviola}. In our case these subsets are defined by the spatial separation of the parties, but more general cases can be considered \cite{barnumviola}.
In particular we study the simplest spatial configuration that allows delocalization, i.e., three spatial modes.

We make the additional assumption that at most one particle is found in any given spatial mode. 
The local Hilbert space of a spatial mode is thus spanned by the two internal degrees of freedom $|\!\uparrow\rangle,|\!\downarrow\rangle$ and $|0\rangle$ representing the absence of a particle. It is the entanglement of these local degrees of freedom of the different spatial modes we consider. In particular we take the maximally entangled states in such a system to be those with maximally mixed one-mode reduced density matrices in analogy with the case of localized particles. For such a state any party has the same probability of finding $|\!\uparrow\rangle$, as $|\!\downarrow\rangle$ or $|0\rangle$ when performing a local measurement.
Furthermore, it is implicit in the definition of the system that the global particle number is conserved.
The Hilbert space of the full system is therefore spanned by the twelve vectors $|ij0\rangle,|0ij\rangle,|i0j\rangle$, where $i,j\in\{\uparrow,\downarrow\}$.

To apply the methods of characterizing the entanglement in this system, such as constructing entanglement invariants or demonstrating Bell-nonlocal correlations the first step is to identify the group of local operations that can be performed. 
For this the global conservation of particle number has important implications.

\subsection{Entanglement and Bell-nonlocality in a system with a superselection rule}

The notion of entanglement depends on the division of the system into parts corresponding to different parties. For three parties A, B, and C we denote such a partitioning of the system by A|B|C. It is assumed that a given party, e.g. party A, can only act on a subset of the degrees of freedom of the system that are locally accessible to A described by the local Hilbert space $H_A$. Alternatively, for a bi-partitioning of the system, e.g., A|BC, it is assumed that the degrees of freedom of both mode B and C are accessible to a party who can perform operations on the Hilbert space $H_B\otimes H_C$. The full Hilbert space of the shared system is described by the tensor product  $H=H_A\otimes H_{B}\otimes H_{C}$ of the individual local Hilbert spaces.
A state, represented by it density matrix $\rho$, is entangled with respect to a partitioning if it cannot be expressed as a {\it separable} state. For example, a three-partite state is separable if it can be expressed as

\begin{eqnarray}\label{sep}\rho=\sum_{k}p_k\rho_A^k\otimes\rho_B^k\otimes\rho_C^k,\end{eqnarray} 
where $\rho_X^k$ is a reduced density matrix for party X, $p_k>0$, and $\sum_k p_k=1$.
The notion of separability clearly depends on the partitioning of the system.
A state that is non-separable for the partition A|B|C may be separable for the partition A|BC. However, a state which is separable for the partitioning A|B|C is separable for any bi-partitioning such as A|BC.

In a system with a conserved quantity, for example a charge, a particle number, or a parity, the set of physically allowed operations is restricted to those that conserve the given quantity. This restriction, called a {\it super-selection rule}, imposes a restriction on the local operations the parties can perform. For an introduction to superselection rules see e.g. Ref. \cite{moretti}.

No local operations that change the conserved quantity are allowed. For example, if the total value of the conserved quantity is the sum of local values, as is the case with particle number, the allowed local operations can not change the local values. In this case we can decompose the Hilbert space $H$ of a system with $n$ spatial modes as a direct sum $H=\oplus_{q_1,q_2,\dots,q_n} H_{q_1,q_2,\dots,q_n}$ where $H_{q_1,q_2,\dots,q_n}$ is the Hilbert space of states with local conserved quantities $q_1,q_2,\dots,q_n$ of the $n$ modes. The allowed local operations cannot transform between the different sectors defined by the different $H_{q_1,q_2,\dots,q_n}$. For example, a local operation $g$ on the part of the system with local conserved quantity $q_1$ can not change the value of $q_1$. Therefore, it can be given a block-diagonal matrix form where each block is a transformation on a subspace of the local Hilbert space corresponding to a given value of $q_1$, i.e.,


\begin{eqnarray}g=\left( \begin{array}{cccccc}
\! {\bf B}_{q_1=n} & \! 0 & 0 & \cdots & 0 & 0\\
\! 0 & {\bf B}_{q_1=n-1} & 0 & \cdots & 0 & 0\\
\! 0 & 0 &  \ddots &  & \vdots & \vdots \\
\! \vdots & \vdots & & \ddots &  0 & 0\\
\! 0 &  0   & \cdots & 0 & {\bf B}_{q_1=1} & 0\\
\! 0 & 0  & \cdots & 0 & 0 & {\bf B}_{q_1=0}\end{array} \!\right),
\end{eqnarray}
where ${\bf B}_{q_1=i}$ is the block containing an operation that preserves the value $q_1=i$ of the local conserved quantity. 
This form of the local operations implies that a coherent superposition $\frac{1}{\sqrt{2}}(|\psi\rangle+|\phi\rangle)$ where $|\psi\rangle$ and $|\phi\rangle$ belong to different sectors, i.e., different values of the local conserved quantities, cannot be distinguished from the incoherent mixture
$\frac{1}{2}|\psi\rangle\langle\psi|+\frac{1}{2}|\phi\rangle\langle\phi|$ by means of local operations.

A super-selection rule can therefore limit the ability of the parties to distinguish different states. In particular it can limit the ability to distinguish entangled from separable states. More precisely, for an entangled state $\rho$ there may exist a separable state $\rho_s$ such that all physically allowed observables by parties A, B, and C yield the same expectation values. Thus we say that $\rho$ is indistinguishable from the separable state $\rho_s$ given the partitioning A|B|C if for all local measurements $Q_{A},Q_{B},Q_{C}$

\begin{eqnarray}
Tr(Q_{A}\times Q_{B}\times Q_{C} \rho)=Tr(Q_{A}\times Q_{B}\times Q_{C} \rho_s).
\end{eqnarray}

However, a state that is indistinguishable from a separable state given the partitioning A|B|C may still have entanglement that is detectable given a more coarse partitioning, e.g. A|BC. This is in contrast to the case without super-selection rules where the entanglement of the system can be fully detected and characterized using the most fine grained partitioning. Thus, in general the operational characterization of entanglement in systems with super-selection rules require all partitionings to be considered. 

If there is no super-selection rule, parties that share a pure entangled state can always find local measurements   
with outcome correlations that cannot be described by a local hidden variable model, i.e., a local description of reality \cite{gisin,popescu3}. Correlations of this kind are called Bell-nonlocal \cite{Bell}. Such nonlocal outcome correlations can be characterized by the violation of a so-called Bell-inequality.
A Bell inequality is a linear expression in the measurement expectation values and their correlators which has a maximum value for Bell-local correlations that can be exceeded by non-local correlations.
For two such parties sharing a system where the local Hilbert spaces are spanned by two degrees of freedom a certain Bell-ineguality, the CHSH-inequality \cite{CHSH}, is violated by any given pure entangled state. However, there exist mixed non-separable states for which all measurement outcome correlations are Bell-local \cite{werner,barret}. Thus, there exist mixed states which can be verified to be entangled by a state tomography but are still Bell-local. 

In the presence of a super-selection rule a pure entangled state may be indistinguishable from a separable state, and thus all correlations are Bell-local. This is in contrast to the case without super-selection rules where such states do not exist.

Note that super-selection rules can allow for quantum information tasks which are not possible otherwise. For example, coding data in the state of the system in such a way that it can only be read by global measurements and is hidden from parties making local measurements \cite{ciracverstraete}.

\subsection{The particle conservation superselection rule}

The global conservation of particle number implies that any local unitary operation made on the system must conserve the local particle number. This restriction on unitary operations implies a restriction also on the Stochastic Local Operations and Classical Communication (SLOCC) \cite{bennett}. Operationally SLOCC is the operations where each party can entangle their modes with a local ancillary system, and subsequently perform a projective measurement on the ancillary system. This is then followed by post selection based on classical communication of the measurement outcomes. The SLOCC satisfying the particle conservation super-selection rule are described by the action of a subgroup $G_{{\mbox{\tiny\itshape SLOCC}}}$ of the special linear group (SL). Locally on each spatial mode the action of $G_{{\mbox{\tiny\itshape SLOCC}}}$ is described by matrices $g$ of the form

\begin{eqnarray}g=\left( \begin{array}{ccc}
\!\! c_{\uparrow\uparrow} & \! c_{\uparrow\downarrow} & 0 \\
\!\! c_{\downarrow\uparrow} & c_{\downarrow\downarrow} & 0 \\
\!\! 0 &\! 0 & c_{00}\end{array} \!\right),\phantom{uu}det(g)=1,
\end{eqnarray}
where $c_{ij}=\langle i|g|j\rangle$ for a the basis $|\!\uparrow\rangle,|\!\downarrow\rangle,|0\rangle$ of the local Hilbert space. The particle conservation super-selection rule is expressed in the block diagonal structure of these matrices, where the $2\times 2$ block corresponds to a particle in the mode and the other block corresponds to the absence of a particle. The group $G_{{\mbox{\tiny\itshape SLOCC}}}$ is a Lie-group \cite{hall} of determinant one matrices. Any element in such a group can be expressed as the matrix exponential of a traceless matrix. The generators of the Lie group is a set of traceless matrices such that any element of the group can be expressed as a matrix exponential of a linear combination of the generators. 

On a given spatial mode, the local action of the group $G_{{\mbox{\tiny\itshape SLOCC}}}$ is generated by the four Gell-Mann matrices \cite{ramond}

\begin{align}\lambda_1=&\left( \begin{array}{ccc}
\! 0 &\,  1 &\, 0 \\
\! 1 &\,  0 &\, 0 \\
\! 0 &\,  0 &\, 0\end{array} \!\right), &\lambda_2=\left( \begin{array}{ccc}
\! 0 &\!  -i & 0 \\
\! i &\!  0 & 0 \\
\! 0 &\! 0 & 0\end{array} \!\right),\nonumber\\
\lambda_3=&\left( \begin{array}{ccc}
\! 1 & \!\! 0 &\! 0 \\
\! 0 & \!\! -1 &\! 0 \\
\! 0 &\!\! 0 &\! 0\end{array} \!\right),  &\lambda_8=\left( \begin{array}{ccc}
\! 1 &\,  0 &\! 0\! \\
\! 0 &\,  1 &\! 0\! \\
\! 0 &\, 0 &\!\! -2\!\end{array} \!\right).
\end{align}
Any element of $G_{{\mbox{\tiny\itshape SLOCC}}}$ can thus be expressed as $e^{\alpha\lambda_1+\beta\lambda_2+\gamma\lambda_3+\delta\lambda_8}\times e^{\epsilon\lambda_1+\zeta\lambda_2+\eta\lambda_3+\theta\lambda_8}\times e^{\kappa\lambda_1+\mu\lambda_2+\nu\lambda_3+\xi\lambda_8}$ for some $\alpha,\dots,\xi \in \mathbb{C}$. 

We can consider the subgroup $G_{123}$ with local action generated by $\lambda_1,\lambda_2,\lambda_3$, i.e., only acting on the local subspace spanned by $|\!\uparrow\rangle$ and $|\!\downarrow\rangle$,  and the subgroup $G_8$ generated by $\lambda_8$. The group $G_8$ contains the only operations involving both $|0\rangle$ and the internal degrees of freedom which are allowed by the super-selection rule. It is clear from the form of the generators that $G_{123}$ and $G_8$ commute. 
Furthermore, the group $G_{{\mbox{\tiny\itshape SU}}}$ of local special unitary operations satisfying the super-selection rule is the subgroup of $G_{{\mbox{\tiny\itshape SLOCC}}}$ with elements satisfying $gg^{\dagger}=I$.

Note that if one of the particles is localized with a party, e.g., party A, the local unitary operations or SLOCC can only act on $|\!\uparrow\rangle$ and $|\!\downarrow\rangle$. Therefore the Gell-Mann matrix $\lambda_8$ is not among the generators for the action on the mode of party A. This type of scenario is not considered in the main text and does not affect the conclusions. See the Appendix for a brief discussion.

The particle conservation super-selection rule leads to a simple relation between detectable entanglement and Bell-nonlocality.
A pure entangled state of two delocalized spin-$\frac{1}{2}$ particles has Bell-nonlocal correlations for a given partitioning if and only if it is distinguishable from a separable state.

\begin{lemma}\label{lemm1}
A pure entangled state of two delocalized spin-$\frac{1}{2}$ particles is Bell-local for a partitioning if and only if it is indistinguishable from a separable state.
\end{lemma}
\begin{proof}

For a partitioning A|BC a general state $|\psi\rangle$ can be written as
\begin{eqnarray}|\psi\rangle=r_1|\!\uparrow\rangle|\theta\rangle+r_2|\!\downarrow\rangle|\xi\rangle+r_3|0\rangle|\phi\rangle,\end{eqnarray} 
where $|\theta\rangle,|\xi\rangle,|\phi\rangle\in H_{BC}$. Consider the projection of the state onto a subspace where one particle is localized with A and one with BC. All states in the subspace spanned by $|\theta\rangle$ and $|\xi\rangle$ contains only one particle. Therefore, within this subspace the operations by the party with access to the modes B and C are not restricted by the super-selection rule. 
Given such operations the system is now isomorphic to a system of two localized spin-$\frac{1}{2}$ particles. Therefore, there exist measurements with Bell-nonlocal outcome correlations if the projection of the state onto this subspace is entangled, i.e., if $|\theta\rangle\neq|\phi\rangle$ and $r_1r_2\neq 0$.
If the state is Bell-local it thus follows that $r_1r_2=0$ or $|\theta\rangle=|\xi\rangle$. Such a state is, up to local unitary operations, of the form $s_1|\!\downarrow\rangle|\varphi\rangle+s_2|0\rangle|\phi\rangle$ and is indistinguishable from a separable state $|s_1|^2|\!\downarrow\rangle\langle\downarrow\!|\otimes|\varphi\rangle\langle\varphi|+|s_2|^2|0\rangle\langle 0|\otimes|\phi\rangle\langle\phi|$. 

For a partitioning A|B|C the question of Bell-locality reduces to the question of Bell locality for measurement performed by the pairs of parties AB, BC, and AC. This follows since only two parties, e.g. A and B, can find particles in any given measurement and measurements performed by C are correlated with the outcomes of A and B in a way that can be described by local hidden variables. Thus the state is Bell-local only if it is a superposition of three product states with the particles in different locations. Up to local unitary transformations these states are of the form $r_1|\!\uparrow\uparrow\!{0}\rangle+r_2|\zeta{0}\!\downarrow\rangle +r_3|0\kappa\chi\rangle$. This state is indistinguishable from the separable state $|r_1|^2\rho_1+|r_2|^2\rho_2+|r_3|^2\rho_3$, where $\rho_1=|\!\uparrow\rangle\langle\uparrow\!|\otimes|\!\uparrow\rangle\langle\uparrow\!|\otimes|{0}\rangle\langle 0|$, $\rho_2=|\zeta\rangle\langle\zeta|\otimes|{0}\rangle\langle 0|\otimes|\!\downarrow\rangle\langle\downarrow\!|$, and $\rho_3=|0\rangle\langle{0}|\otimes|\kappa\rangle\langle\kappa|\otimes|\chi\rangle\langle\chi|$.

Finally, a state that is indistinguishable from a separable state is trivially Bell-local.
\end{proof}

The ability to detect entanglement in a pure state is thus equivalent to the ability to detect Bell-nonlocal correlations.

\subsection{Entanglement classification using invariants}

A characterization of the entanglement of a system can be made in terms of quantities and parameters specific to a particular physical realization, but here we will focus on realization independent properties, i.e., intrinsic properties of the system. 
In particular, we consider
the entanglement between different parts of the system and the local unitary symmetries of a system. The significance of the local unitary symmetries is that they identify certain exceptional types of entanglement that often have interesting properties.

A general tool for entanglement classification is the polynomial entanglement invariants \cite{linden,grassl,wootters,popescu2,coffman,sudbery,luquethibon1,Osterlohsiewert2005}. 
A polynomial invariant of $G_{{\mbox{\tiny\itshape SLOCC}}}$ ($G_{{\mbox{\tiny\itshape SU}}}$) is a polynomial in the coefficients of the state vector (in the coefficients and their complex conjugates) which is invariant under the action of $G_{{\mbox{\tiny\itshape SLOCC}}}$ ($G_{{\mbox{\tiny\itshape SU}}}$). 
Two states are said to have the same type of entanglement with respect to $G_{{\mbox{\tiny\itshape SLOCC}}}$ ($G_{{\mbox{\tiny\itshape SU}}}$) if they can be converted into each other by $G_{{\mbox{\tiny\itshape SLOCC}}}$ ($G_{{\mbox{\tiny\itshape SU}}}$). If this is the case any ratio between two $G_{{\mbox{\tiny\itshape SLOCC}}}$ invariants ($G_{{\mbox{\tiny\itshape SU}}}$ invariants) take the same value for the two states.
In this sense the invariants under $G_{{\mbox{\tiny\itshape SLOCC}}}$ ($G_{{\mbox{\tiny\itshape SU}}}$) serve as coordinates on the space of $G_{{\mbox{\tiny\itshape SLOCC}}}$ ($G_{{\mbox{\tiny\itshape SU}}}$) inter-convertibility classes, i.e., different types of entanglement. A subset of the invariants with the property that all invariants can be constructed algebraically from the subset is called a set of generators \cite{olver}. A set of generators is thus sufficient to describe all entanglement types in the system.

Suitably chosen the generators of the invariants correspond to physically relevant properties of the system that depend on the entanglement.
In particular a choice can be made such that the set of states where one or several generators take the value zero correspond to an intrinsic symmetry \cite{johansson14,luna}. This choice also captures exceptional dimensions \cite{carteret} and topologies \cite{johanssonosterloh14} of the local unitary inter-convertibility classes of states. 

A SLOCC invariant can also be used to construct a function that is non-increasing under SLOCC, an {\it entanglement monotone} \cite{vidal,verstraete2003}. If $I$ is a SLOCC invariant of degree $k$, the function $|{I}|^{2/k}$ is an entanglement monotone \cite{verstraete2003}. To understand the meaning of an entanglement monotone we note the two following properties of SLOCC operations. A SLOCC operation, when represented by a determinant one matrix, does not in general preserve the norm of the state vector. If the state vector is normalized the value of $I$ is therefore in general not preserved by a SLOCC operation. Moreover, a SLOCC operation can in general not be performed with unit probability of success. The outcome is a statistical mixture of the desired outcome state and some other state(s), and the desired operation is realized only after postselection to remove the undesired outcomes.

When taking the outcomes of the unsuccessful attempts to perform the operation into account, and properly normalizing the outcome state vectors, the value of $I$ is in general not constant. However it was shown in Ref. \cite{verstraete2003} that the function $|{I}|^{2/k}$ does not increase when averaged over both successful and unsuccessful attempts.
Since this function does not increase under SLOCC it is called an entanglement monotone and quantifies a property of the system that cannot be created by SLOCC. In a system without superselection rules this directly implies that it quantifies a property of the entanglement since all non-entangled states can be created by SLOCC. In this case the monotone thus serves as a measure, of the particular type of entanglement distinguished by the invariant it was constructed from. In the presence of a particle conservation superselection rule the implication that the monotone measures a property of the entanglement is not direct because the set of SLOCC operations is restricted. However, it can be shown that in the present case a nonzero value of an entanglement monotone implies that Bell-nonlocal correlations can be observed across all bi-partitions.

\begin{theorem}\label{t1}
If a state $|\psi\rangle$ is Bell-local over a bipartition, there exist no maximally entangled state to which $|\psi\rangle$ can be converted by SLOCC and no entanglement monotone takes a non-zero value for $|\psi\rangle$. Moreover, local special unitary operations can induce an arbitrary global phase shift of the state vector $|\psi\rangle$. Thus, if an entanglement monotone takes a non-zero value the state is Bell-nonlocal over every bi-partition.
\end{theorem}
\begin{proof}

As described in the proof of Lemma \ref{lemm1}, if a state $|\psi\rangle$ is Bell-local over a bipartition A|BC there exist a basis such that
$|\psi\rangle=r_1|\!\uparrow\rangle|\theta\rangle+r_2|0\rangle|\phi\rangle$. For a state of this form the determinant one matrix $e^{-3\alpha\lambda_3+\alpha\lambda_8}\times{I}\times{I}$ changes the norm of the state vector by a factor of $e^{-2\alpha}$. By making $\alpha$ sufficiently large we can bring the norm arbitrarily close to zero. This implies that the set of states which are inter-convertible with $|\psi\rangle$ does not contain any maximally entangled states \cite{kempfness,verstraete2003}.
The form of the state $|\psi\rangle$ also implies that
$e^{2i\alpha}e^{-3i\alpha\lambda_3+i\alpha\lambda_8}\times{I}\times{I}$ is an element of the stabilizer group. This means that under the special unitary operation $e^{-3i\alpha\lambda_3+i\alpha\lambda_8}\times{I}\times{I}$ the state accumulates a phase-factor $e^{-2i\alpha}$.

Vice versa, if there exist an entanglement monotone which takes a non-zero value for a state, the state cannot have the form $|\psi\rangle=|\!\uparrow\rangle|\theta\rangle+|0\rangle|\phi\rangle$ for any bipartition, since no homogeneous polynomial is invariant under $e^{-3\alpha\lambda_3+\alpha\lambda_8}\times{I}\times{I}$. Thus, the state must be Bell-nonlocal over every bipartition.

\end{proof}
Thus a monotone serves as a measure of detectable entanglement, and in particular of the type of detectable entanglement distinguished by the invariant it was constructed from. 


A maximally entangled state is defined as a pure state for which all the single party reduced density matrices are maximally mixed.
If a given state $|\psi\rangle$ and some maximally entangled state can be converted into each other by SLOCC (with a non-zero probability of success), there exist at least one entanglement monotone which takes a non-zero value for $|\psi\rangle$. Furthermore, on the SLOCC inter-convertibility class of $|\psi\rangle$ the monotone takes its maximal value on the maximally entangled states \cite{kempfness,verstraete2003}.
For the states that can not be transformed by SLOCC to a maximally entangled state, a non-zero entanglement monotone still exist if the state can be brought arbitrarily close to a maximally entangled state by SLOCC. 

Bell-nonlocality across all bipartitions is a requirement for entanglement measures as described by Theorem \ref{t1} and therefore also for maximal entanglement.
Maximally entangled states and more generally states for which some entanglement monotone takes a non-zero value are Bell-nonlocal for every bipartition of the system.
However, maximal entanglement or non-zero entanglement monotones do not imply Bell-nonlocality over A|B|C.

Next we provide the motivation why delocalization over three modes is considered.
In a system of two delocalized spin-$\frac{1}{2}$ particles where at most one particle can occupy a spatial mode, SLOCC invariants can only exist if the number of spatial modes is three.

\begin{theorem}
\label{thee}
There exist no entanglement measures or maximally entangled states for two delocalized spin-$\frac{1}{2}$ particles, where at most one particle occupy a spatial mode, if the number of spatial modes is larger than three.
\end{theorem}
\begin{proof}
Consider two delocalized particles in a system of $n$ spatial modes and the matrix $g$

\begin{eqnarray}g=\left( \begin{array}{ccc}
\!\! re^{i\phi} & \! 0 & 0 \\
\!\! 0 & re^{i\phi} & 0 \\
\!\! 0 &\! 0 & r^{-2}e^{-2i\phi}\end{array} \!\right).
\end{eqnarray}
The action of the $G_{{\mbox{\tiny\itshape SLOCC}}}$ group element $g^{\times{n}}$ on a basis vector depends only on the number of occupied spatial modes. Therefore this action is a scaling factor $r^{-2(n-3)}e^{-2(n-3)i\phi}$ of the state vector. 

It follows that for $r\neq{1}$ it is impossible to construct polynomials in the state vector coefficients that are invariant under this action unless $n=3$. 
\end{proof}
Thus, a description of the maximally entangled states in a system with three spatial modes is an exhaustive description of maximal entanglement for two delocalized spin-$\frac{1}{2}$ particles.

\subsection{SLOCC invariants}
\label{slocc}
We now describe the invariants of $G_{{\mbox{\tiny\itshape SLOCC}}}$ and how to find them. 
Our method to construct the $G_{{\mbox{\tiny\itshape SLOCC}}}$ invariants is to find the algebra of polynomial $G_{123}$-invariants and then select the subalgebra of $G_8$ invariants. Since $G_{123}$ and $G_8$ commute this subalgebra contains all the $G_{{\mbox{\tiny\itshape SLOCC}}}$ invariants.
Furthermore, $G_{123}$ is isomorphic to $SL(2)^{\times 3}$ and therefore we can use Cayley's Omega Process (see e.g. \cite{olver}) to find the $G_{123}$ invariants. This procedure is described in the Appendix.

It is found that the invariants of $G_{{\mbox{\tiny\itshape SLOCC}}}$ are generated by the two polynomials

\begin{align}\label{i1}
 I_1 =&(m_{\uparrow\uparrow{0}}m_{\downarrow\downarrow{0}}-m_{\uparrow\downarrow{0}}m_{\downarrow\uparrow{0}})\nonumber\\
 \times &(m_{0\uparrow\uparrow}m_{0\downarrow\downarrow}-m_{0\uparrow\downarrow}m_{0\downarrow\uparrow})\nonumber\\
 \times &(m_{\uparrow{0}\uparrow}m_{\downarrow{0}\downarrow}-m_{\uparrow{0}\downarrow}m_{\downarrow{0}\uparrow}),
\end{align}
and
\begin{align}\label{i4}
I_{2}=&m_{\uparrow\uparrow{0}}(m_{{0}\downarrow\downarrow}m_{\downarrow{0}\uparrow}-m_{{0}\downarrow\uparrow}m_{\downarrow{0}\downarrow})\nonumber\\+&m_{\uparrow\downarrow{0}}(m_{0\uparrow\uparrow}m_{\downarrow{0}\downarrow}-m_{{0}\uparrow\downarrow}m_{\downarrow{0}\uparrow})\nonumber\\+&m_{\downarrow\uparrow{0}}(m_{{0}\downarrow\uparrow}m_{\uparrow{0}\downarrow}-m_{{0}\downarrow\downarrow}m_{\uparrow{0}\uparrow})\nonumber\\+&m_{\downarrow\downarrow{0}}(m_{{0}\uparrow\downarrow}m_{\uparrow{0}\uparrow}-m_{0\uparrow\uparrow}m_{\uparrow{0}\downarrow}),
\end{align}
where $m_{ijk}$ is the probability amplitude of $|ijk\rangle$ for $i,j,k\in\{\uparrow,\downarrow,0\}$. The corresponding monotones are $|I_1|^{1/3}$ and $|I_2|^{2/3}$, respectively.
The polynomial $I_1$ is the product of the concurrence polynomials \cite{wootters} corresponding to the three pairs of spatial modes. The individual concurrence polynomials are invariant under $G_{123}$ but not under $G_8$. However, the modulus squared of a concurrence polynomial is a local unitary invariant. We denote these invariants $I_{AB},I_{BC}$, and $I_{AC}$ where

\begin{align}
I_{AB}=&|m_{\downarrow\uparrow{0}}m_{\uparrow\uparrow{0}}-m_{\uparrow\downarrow{0}}m_{\downarrow\uparrow{0}}|^2,\nonumber\\
I_{BC}=&|m_{{0}\downarrow\downarrow}m_{{0}\uparrow\uparrow}-m_{0\uparrow\downarrow}m_{{0}\downarrow\uparrow}|^2,\nonumber\\
I_{AC}=&|m_{\downarrow{0}\downarrow}m_{\uparrow{0}\uparrow}-m_{\uparrow{0}\downarrow}m_{\downarrow{0}\uparrow}|^2.
\end{align}
The invariant $I_{AB}$ takes a non-zero value only if the projection of the state onto a subspace where the particles are localized with A and B is an entangled state of two localized particles. 
This is equivalent to the existence of Bell-nonlocal correlations between A and B as described in the proof of Lemma \ref{lemm1}. 
The invariants $I_{BC}$ and $I_{AC}$ have the same property for B and C, and A and C, respectively.
The concurrence polynomials thus describe the detectable entanglement between parties in a tri-partition of the system.
The invariant $I_1$ thus takes a non-zero value only if each projection of the state, onto a subspace where both particles are localized, is an entangled state of two localized particles. This is equivalent to each pair of parties being able to observe Bell-nonlocal correlations.

Both $I_1$ and $I_2$ measure genuine tri-partite entanglement in the sense that they take zero value if there is any bipartition over which the state is a product. 
Note that $I_1$ is symmetric under the interchange of two parties while $I_2$ is antisymmetric.

\section{The entangled states of two spin-$\frac{1}{2}$ particles and their Bell-nonlocality}
\label{sec3}
We now characterize the the entanglement of the system of two delocalized particles in terms of detectable entanglement, i.e., Bell nonlocality for different partitions and different parties. Since the Hilbert space of the system is spanned by twelve basis vectors the number of real parameters that are needed to describe the set of state vectors is 24. However, the entanglement of a state is invariant under local unitary transformations. Hence, by removing these local degrees of freedom one can more clearly see the degrees of freedom relevant for the entanglement \cite{acin,higushi,acin2}. To facilitate the characterization of the entanglement in the system we therefore construct a parametrized set of states such that any state can be transformed by a local unitary operation to a state in this set.
The set has twelve parameters and is given by

\begin{align}\label{onewee00}
&r_1|\!\uparrow\uparrow\!{0}\rangle+r_2|\!\downarrow\downarrow\!{0}\rangle+r_3|\!\downarrow\!{0}\!\downarrow\rangle+r_4|\!\uparrow\!{0}\!\uparrow\rangle+r_5e^{i\phi}|\!\uparrow\!{0}\!\downarrow\rangle\nonumber\\+&r_6|0\!\downarrow\downarrow\rangle+r_7|0\!\uparrow\uparrow\rangle
+r_8e^{i\varphi}|0\!\downarrow\uparrow\rangle+r_9e^{i\theta}|0\!\uparrow\downarrow\rangle,
\end{align}
where $\phi,\varphi,\theta$ are real numbers and $r_1,\dots,r_9$ are real and non-negative numbers satisfying $\sum_{j=1}^9r_j^2=1$. The expression in Eq. (\ref{onewee00}) can be derived from a general parametrization of a state vector by a straightforward use of $G_{{\mbox{\tiny\itshape SU}}}$ operations to reduce the number of basis vectors from twelve to nine. First two unitary operations on the mode of $A$ and the mode of $B$ respectively to bring the part of the state vector in the subspace where the particles are with $A$ and $B$ to the Schmidt form \cite{schmidt} $r_1|\!\uparrow\uparrow\!{0}\rangle+r_2|\!\downarrow\downarrow\!{0}\rangle$, followed by a unitary operation on the mode of $C$ to bring the part of the state vector in the subspace where the particles are with $A$ and $C$ to the form $r_3|\!\downarrow\!{0}\!\downarrow\rangle+r_4|\!\uparrow\!{0}\!\uparrow\rangle+r_5e^{i\phi}|\!\uparrow\!{0}\!\downarrow\rangle$, i.e., reducing the number of basis vector needed for this subspace by one.
This is then followed by an elimination of two relative phase factors of the remaining basis vectors using the diagonal subgroup of $G_{{\mbox{\tiny\itshape SU}}}$ generated by $\lambda_3$ and $\lambda_8$ on each of the modes.

As described in Sect. \ref{slocc} measurements by two parties, e.g., A and B, can reveal Bell-nonlocality if and only if the concurrence polynomial \cite{wootters} $m_{\uparrow\uparrow{0}}m_{\downarrow\downarrow{0}}-m_{\uparrow\downarrow{0}}m_{\downarrow\uparrow{0}}$ takes a non-zero value. 
For the parametrization in 
Eq. (\ref{onewee00}) this means $r_1r_2\neq 0$. Likewise, Bell-nonlocality for BC and AC requires $r_6r_7-r_8r_9e^{i(\varphi+\theta)}\neq0$, and $r_3r_4\neq0$, respectively. Since these constraints are satisfied only by a low dimensional subset of states, a generic state is nonlocal for each pair of parties.
Moreover, if a state is Bell-local for each of the pairs AB, BC, and AC it is Bell-local for the tri-partition A|B|C as described in the proof of Lemma \ref{lemm1}, and the invariant $I_1$ which is the product of three concurrence polynomials thus takes the value zero.

Next we consider Bell-nonlocality for bi-partitions of the system. If a state is Bell-nonlocal for measurements performed by A and B it is trivially Bell-nonlocal for the bi-partitions A|BC and B|AC. However, Bell-nonlocality across bipartitions does not imply anything about correlations between any pair of parties, e.g. A and B. States that are Bell-local over A|B|C can still be nonlocal over bipartitions. 



\subsection{Entangled states that are Bell-local for three parties}
\label{kk}
We now describe the set of states that are Bell-local for the tripartition A|B|C.
This set can be further divided into subsets based on for which bipartitions entanglement can be detected. The presence of detectable entanglement, i.e., Bell-nonlocality, across a given bipartition, e.g. A|BC, given that the tripartiton A|B|C is Bell-local, is described by a non-zero value of a particular unitary invariant $I_{A|BC}$. The three polynomial invariants of this type are 
 
\begin{align}
I_{A|BC}=&\nonumber\\=|m_{\downarrow\uparrow{0}}&m_{\uparrow{0}\uparrow}-m_{\uparrow\uparrow{0}}m_{\downarrow{0}\uparrow}|^2+|m_{\downarrow\uparrow{0}}m_{\uparrow{0}\downarrow}-m_{\uparrow\uparrow{0}}m_{\downarrow{0}\downarrow}|^2\nonumber\\+|m_{\downarrow\downarrow{0}}&m_{\uparrow{0}\uparrow}-m_{\uparrow\downarrow{0}}m_{\downarrow{0}\uparrow}|^2+|m_{\downarrow\downarrow{0}}m_{\uparrow{0}\downarrow}-m_{\uparrow\downarrow{0}}m_{\downarrow{0}\downarrow}|^2,\nonumber\\
I_{C|AB}=&\nonumber\\=|m_{{0}\uparrow\downarrow}&m_{\uparrow{0}\uparrow}-m_{0\uparrow\uparrow}m_{\uparrow{0}\downarrow}|^2+|m_{{0}\uparrow\downarrow}m_{\downarrow{0}\uparrow}-m_{0\uparrow\uparrow}m_{\downarrow{0}\downarrow}|^2\nonumber\\+|m_{{0}\downarrow\downarrow}&m_{\uparrow{0}\uparrow}-m_{{0}\downarrow\uparrow}m_{\uparrow{0}\downarrow}|^2+|m_{{0}\downarrow\downarrow}m_{\downarrow{0}\uparrow}-m_{{0}\downarrow\uparrow}m_{\downarrow{0}\downarrow}|^2,\nonumber\\
I_{B|AC}=&\nonumber\\=|m_{\uparrow\downarrow{0}}&m_{0\uparrow\uparrow}-m_{\uparrow\uparrow{0}}m_{{0}\downarrow\uparrow}|^2+|m_{\uparrow\downarrow{0}}m_{{0}\uparrow\downarrow}-m_{\uparrow\uparrow{0}}m_{{0}\downarrow\downarrow}|^2\nonumber\\+|m_{\downarrow\downarrow{0}}&m_{0\uparrow\uparrow}-m_{\downarrow\uparrow{0}}m_{{0}\downarrow\uparrow}|^2+|m_{\downarrow\downarrow{0}}m_{{0}\uparrow\downarrow}-m_{\downarrow\uparrow{0}}m_{{0}\downarrow\downarrow}|^2.\nonumber\\
\end{align}
The construction of these invariants is described in the Appendix.

The states for which no bipartition feature nonlocal correlations are equivalent up to local unitary operations to a state of the form

\begin{eqnarray}\label{onewee011}
S_{{\mbox{\tiny\itshape L:A|BC,B|AC,C|AB}}}=r_1|\!\downarrow\downarrow\!{0}\rangle+r_2|\!\downarrow\!{0}\!\downarrow\rangle+r_3|0\!\downarrow\downarrow\rangle,
\end{eqnarray} 
where $r_1,r_2,r_3$ are real and non-negative satisfying $\sum_{j=1}^3r_j^2=1$.
These are simply the states where the particles are delocalized but no entanglement exist between the internal degrees of freedom of the different spatial modes or between internal degrees of freedom and spatial location. 
The invariants are accordingly $I_1=I_2=I_{A|BC}=I_{B|AC}=I_{C|AB}=0$ for this class of states.

Next we consider the states where only the bipartition A|BC features nonlocal correlations. For these states the invariant $I_{A|BC}\neq 0$ while $I_{B|AC}=I_{C|AB}=0$. Up to local unitary operations such a state is of the form

\begin{eqnarray}\label{onewee012}
S^{ {\mbox{\tiny\itshape NL:A|BC}}}_ {{\mbox{\tiny\itshape L:B|AC,C|AB}}}=r_1|\!\downarrow\downarrow\!{0}\rangle+r_2|\!\uparrow\!{0}\!\uparrow\rangle+r_3|0\!\downarrow\uparrow\rangle,
\end{eqnarray}
where $r_1,r_2>0$, $r_3\geq{0}$, and $\sum_{j=1}^3r_j^2=1$. 
A simple example of a state in this class is two maximally entangled photons where one of the photons have passed through a spin polarizing beam splitter $\frac{1}{\sqrt{2}}(|\!\downarrow\downarrow\!{0}\rangle+|\!\uparrow\!{0}\!\uparrow\rangle)$. 
States with nonlocal correlations across only B|AC or C|AB can be obtained from Eq. (\ref{onewee012}) by permuting the spatial modes.

States which are Bell-nonlocal over A|BC and C|AB but not over B|AC are up to local unitary operations equivalent to
\begin{align}\label{onewee013}
S^{{\mbox{\tiny\itshape NL:A|BC,C|AB}}}_{{\mbox{\tiny\itshape L:B|AC}}}=r_1|\!\downarrow\downarrow\!{0}\rangle+r_2|\!\uparrow\!{0}\!\uparrow\rangle+r_3|0\!\downarrow\downarrow\rangle
+r_4e^{i\phi}|0\!\downarrow\uparrow\rangle,
\end{align}
with the constraints $r_1,r_2,r_3>0$, $r_4\geq 0$, and $\sum_{j=1}^4r_j^2=1$. For this class of states $I_{A|BC}=r_1^2r_2^2$ and $I_{C|AB}=r_2^2r_3^2$ while $I_{B|AC}=0$. 
The states with nonlocal correlations over A|BC and B|AC or B|AC and C|AB can be obtained by permuting the spatial modes.

Finally we consider the states where all bi-partitions are Bell-nonlocal. For any state that is Bell-local over the tri-partition A|B|C but nonlocal over each bipartition A|BC, B|AC and C|AB, the invariant $I_2\neq0$. Any such state belongs to the set

\begin{align}\label{onewee014}
S^{{\mbox{\tiny\itshape NL:A|BC,B|AC,C|AB}}}=&
r_1|\!\downarrow\downarrow\!{0}\rangle+r_2|\!\uparrow\!{0}\!\uparrow\rangle+r_3|0\!\downarrow\downarrow\rangle+r_4|0\!\uparrow\uparrow\rangle\nonumber\\
+&r_5e^{i\theta}|0\!\uparrow\downarrow\rangle+\frac{r_3r_4}{r_5}e^{-i\theta}|0\!\downarrow\uparrow\rangle,
\end{align}
where $r_1,r_2,r_5>0$, $r_3,r_4\geq0$, and $r_1^2+r_2^2+r_3^2+r_4^2+r_5^2+(r_3r_4/r_5)^2=1$. For this set of states $I_2=r_1r_2r_5e^{i\theta}$.

Note that since the polynomials $I_{A|BC}$, $I_{B|AC}$, and $I_{C|AB}$ are not invariant under SLOCC they cannot be used to construct measures for their corresponding types of entanglement. By Theorem (\ref{t1}) no such measures exist for the states in families of Eq. (\ref{onewee011}), (\ref{onewee012}), and (\ref{onewee013}).

\subsection{The intrinsic symmetry of a generic entangled state}

The delocalization of the two spin-$\frac{1}{2}$ particles gives rise to a local unitary symmetry of the system which is qualitatively different from symmetries of localized spin-$\frac{1}{2}$ particles. To see this we consider the intrinsic symmetry of a generic entangled state, i.e., any state in the family in Eq. (\ref{onewee00}) except for some lower dimensional subsets of states. These lower dimensional subsets, with a higher degree of symmetry, correspond to one or more constraints on the parameters $\phi,\theta,\varphi$ and $r_1,\dots,r_9$ in Eq. (\ref{onewee00}).
The generic intrinsic symmetry is described by the stabilizer group with elements

\begin{align}\label{jule1}e^{i\pi m\lambda_3+i\alpha\lambda_8}
\times {e^{i(\alpha+\frac{\pi}{3}m)\lambda_8}}
\times {e^{i(\alpha-\frac{\pi}{3}m)\lambda_8}},\end{align}
where $\alpha$ is a real valued parameter and $m$ is an integer. Note that the relation $e^{\frac{2\pi i}{3}\lambda_8}=e^{\frac{2\pi i}{3}}I$ leads to an ambiguity in the representation of the group elements in  Eq. (\ref{jule1}). 

The existence of a generic non-trivial symmetry can be contrasted with the case of three or more localized spin-$\frac{1}{2}$ particles where no non-trivial generic symmetry exist \cite{popescu2}. A bipatite system of localized spin-$\frac{1}{2}$ particles on the other hand has a generic non-trivial symmetry. However, the symmetry in Eq. (\ref{jule1}) is not related to any symmetry of localized particles, but is a consequence of the delocalization. This can be seen from the fact that the only element in the stabilizer group in Eq. (\ref{jule1}) which does not involve $\lambda_8$ is the identity element.

An example of a subset of states with a non-generic symmetry is the set $S^{{\mbox{\tiny\itshape NL:A|BC,C|AB}}}_{{\mbox{\tiny\itshape L:B|AC}}}$ in Eq. (\ref{onewee013}). These are the states which are Bell local for A|B|C and Bell-local over B|AC but nonlocal for the other bi-partitions. The generic symmetry of this subset is described by the larger stabilizer group

\begin{eqnarray}\label{sym2}e^{i\alpha}e^{i\beta\lambda_{3}+i\gamma\lambda_8}
\times{e^{i\frac{3}{2}\alpha\lambda_{3}+i(\frac{1}{3}\beta+\gamma+\frac{1}{2}\alpha)\lambda_8}}
\times{e^{+i(-\frac{1}{3}\beta+\gamma)\lambda_8}},\nonumber\\\end{eqnarray}
where $\alpha,\beta,\gamma$ are real valued parameters.
This specific configuration of Bell-locality and nonlocality for the different partitions thus coincides with a higher degree of symmetry. Similarly, the other non-generic Bell-nonlocality configurations in Eq. (\ref{onewee011}), (\ref{onewee012}), and (\ref{onewee014}) correspond to larger symmetry groups.

A symmetry of the type in Eq. (\ref{sym2}) which includes a phase factor $e^{i\alpha}$ implies that local special unitary operations can induce an arbitrary phase shift of the state vector. For a state with the generic symmetry on the other hand, local special unitary operations can only induce a discrete set $\{\pm 1,\pm e^{i\frac{2\pi}{3}}\}$ of phase factors. 
This is an example of how the Bell-locality across a bipartition is reflected in the symmetry properties of the state as described in Theorem (\ref{t1}). Furthermore, the discrete set of phases is related to the invariant $I_2$ since it is a polynomial of degree 3 and therefore only invariant under global phase shifts which are multiples of $2\pi/3$. 

When local unitary evolution of a state can only induce a discrete set of phase shifts these are called topological phases \cite{milmanmosseri,oxmankhoury,johansson12}. The name comes from the relation to the topology of inter-convertibility classes of states under local special unitary operations.
These topological phases can be observed as an interference effect in systems undergoing special unitary evolution \cite{souza,johanssonkhoury13}. 

For any number of localized spin-$\frac{1}{2}$ particles the topological phases of a generic state are only $\pm\pi$. 
Any other set of topological phases of the form $\{\frac{m\pi}{k}\}$ for some given integer $k\neq{1}$ and integers $m$ exist only in lower dimensional subsets of states \cite{johansson12}. The generic fractional topological phase $2\pi/3$ is thus a qualitative feature related to the generic symmetry, and caused by the delocalization of the particles, which never appears in systems of localized spin-$\frac{1}{2}$ particles.

\section{The maximally entangled states}
\label{sec4}
The set of maximally entangled states can be described up to local unitary operations by the two families of states, $S_{1}(r)$ and $S_{2}(r,\theta)$. These can be constructed from the states in Eq. (\ref{onewee00}) by imposing the conditions that all reduced single party density matrices are maximally mixed. The first family $S_{1}(r)$ can be understood as an equal superposition of  three different maximally entangled states of two localized particles on the pairs of spatial modes A and B, A and C, and B and C respectively. The family $S_{1}(r)$ can be expressed as

\begin{align}\label{onewee}
S_{1}(r)=&\frac{1}{\sqrt{6}}(|\!\uparrow\uparrow\!{0}\rangle+|\!\downarrow\downarrow\!{0}\rangle+|\!\downarrow\!{0}\!\downarrow\rangle+|\!\uparrow\!{0}\!\uparrow\rangle)\nonumber\\+&\sqrt{1/6-r^2}\left(|0\!\downarrow\downarrow\rangle+|0\!\uparrow\uparrow\rangle\right)\nonumber\\
+&r(|0\!\downarrow\uparrow\rangle-|0\!\uparrow\downarrow\rangle),
\end{align}
where $0\leq r\leq{1/\sqrt{6}}$. 
If a state in this class is projected onto a subspace where the two particles are localized in a pair of spatial modes, e.g. A and B, the resulting state is a maximally entangled state of two localized particles. 

With this parametrization $I_1=6^{-3}$ and $I_2=r/3$.
The invariant $I_1$, is constant since it is the product of the three concurrence polynomials and each such polynomial is a measure of the bipartite entanglement of two localized particles. For the parametrization in Eq. (\ref{onewee}) the invariant $I_2$ is proportional to the amplitude of the singlet state $|0\!\downarrow\uparrow\rangle-|0\!\uparrow\downarrow\rangle$ for the spatial modes B and C. A maximally entangled state of two localized particles can be transformed into any other such maximally entangled state by local unitary transformations, but the three pairs of superposed such states in Eq. (\ref{onewee}) cannot be transformed independently of each other. The degree to which the three superposed entangled particle-pairs can be simultaneously transformed to the same bipartite state on their respective pair of spatial modes is described by the parameter $r$.  
This "relative orientation" of the states of the three entangled pairs, and the corresponding bi-partite correlations, is quantified by the invariant $I_2$. The three pairs can be oriented in the same way only if $I_2=0$.

The second family $S_{2}(r,\theta)$ is given by

\begin{align}\label{onewee2}
S_{2}(r,\theta)=&r(|\!\uparrow\uparrow\!{0}\rangle+|\!\downarrow\!{0}\!\downarrow\rangle+e^{i\theta}|0\!\downarrow\uparrow\rangle)\nonumber\\
+&\sqrt{\frac{1}{3}-r^2}(|\!\downarrow\downarrow\!{0}\rangle+|\!\uparrow\!{0}\!\uparrow\rangle +|0\!\uparrow\downarrow\rangle),
\end{align}
where $0\leq r<{1/\sqrt{3}}$ and $0\leq\theta\leq 2\pi$. The two invariants take the values  
$I_1=r^3(1/3-r^2)^{3/2}e^{i\theta}$ and $I_2=r^3e^{i\theta}-(1/3-r^2)^{3/2}$, respectively.
For $0<r<{1/\sqrt{3}}$ the states in this family can be described as super-positions of non-maximally entangled pairs of localized particles. As opposed to states of the first family it is not maximal entanglement of these pairs but the precise relation between the states of the pairs that makes the state in Eq. (\ref{onewee2}) maximally entangled. Moreover, for $r=0$ the state cannot anymore be described in terms of bipartite entanglements between localized particles.
This value of $r$ correspond to the unique maximally entangled state, up to local unitary transformations, which is Bell local for A|B|C. This is also the only maximally entangled state for which $I_1=0$. 

The zero points of the two invariants $I_1$ and $I_2$ on the set of maximally entangled states corresponds to two states with special entanglement properties and a high degree of symmetry. We consider these two states in more detail.

The maximally entangled state distinguished by $I_2=0$ in the first family is the state for which all the superposed maximally entangled pairs of two localized particles can be given the same orientation.

\begin{eqnarray}\label{twowee}
|\psi_{1}\rangle\!=\!\frac{1}{\sqrt{6}}(|\!\uparrow\uparrow\!{0}\rangle\!+\!|\!\downarrow\downarrow\!{0}\rangle\!+\!|\!\downarrow\!{0}\!\downarrow\rangle\!+\!|\!\uparrow\!{0}\!\uparrow\rangle\!+\!|0\!\downarrow\downarrow\rangle\!+\!|0\!\uparrow\uparrow\rangle).\nonumber\\
\end{eqnarray}
This is the only state that can be understood as the delocalization of a maximally entangled pair of two spin-$\frac{1}{2}$ particles where the spatial degree of freedom has not been entangled with the orientation of the pair.
The state has a high degree of symmetry described by the non-Abelian stabilizer group defined by the elements 
\begin{align}\label{jule2}&e^{i\frac{\pi}{3}(k+l+m+n+p+q)}e^{i\frac{\pi}{2}[k-l-m+n+p+q]\lambda_3+i\alpha\lambda_8}\nonumber\\
&\times {e^{i\frac{\pi}{2}[k-l+m-n-p+q]\lambda_{3}+i[\alpha+\frac{\pi}{3}(-m-n+p+q)]\lambda_8}}\nonumber\\
&\times {e^{i\frac{\pi}{2}[-k+l-m+n-p+q]\lambda_{3}+i[\alpha+\frac{\pi}{3}(-k-l+p+q)]\lambda_8}},\nonumber\\
&e^{i{\pi}}e^{i\frac{\pi}{2}\lambda_1}
\times{e^{i\frac{\pi}{2}\lambda_1}}
\times{e^{i\frac{\pi}{2}\lambda_1}},\nonumber\\
\phantom{o}\mathrm{and,}\phantom{o} 
&e^{i\beta\lambda_2}
\times{e^{i\beta\lambda_2}}
\times{e^{i\beta\lambda_2}},\end{align}
where $\alpha,\beta$ are real numbers and $k,l,m,n,p,q$ are integers.
This symmetry is distinguished from all other symmetries by a fractional topological phase $\pi/3$ instead of the generic $2\pi/3$.

In the second family we can consider the unique maximally entangled state which is Bell local over the tri-partition A|B|C, i.e., for which $I_1=0$.

\begin{eqnarray}\label{trewee}
|\psi_{2}\rangle=\frac{1}{\sqrt{3}}(|\!\uparrow\uparrow\!{0}\rangle+|\!\downarrow\!{0}\!\downarrow\rangle+|0\!\downarrow\uparrow\rangle).
\end{eqnarray}
This state has an intrinsic symmetry described by the local unitary stabilizer group with elements

\begin{align}\label{jule}{e^{i\alpha\lambda_3+i\beta\lambda_8}}&\times{e^{i\gamma\lambda_{3}
+\frac{1}{3}i(3\beta-\alpha+\gamma+\delta)\lambda_8}}\nonumber\\
&\times{e^{i\delta\lambda_{3}+\frac{1}{3}i(3\beta+\alpha+2\gamma+\delta)\lambda_8}},\end{align}
where $\alpha,\beta,\gamma,\delta$ are a real valued parameters. 
Note that the trivial factors $e^{\frac{2i\pi}{3}}e^{-\frac{2i\pi}{3}\lambda_8}=I$ have been removed.

In conclusion, we can see that maximally entangled states of two delocalized spin-$\frac{1}{2}$ particles are in general not simply entangled pairs of particles that have been delocalized. Only one maximally entangled state, up to local unitary operations, can be described this way. In general the spin degrees of freedom are entangled with the spatial location in more subtle ways, and in the extreme case this leads to a maximally entangled state that is Bell-local for three parties performing local measurements.

\section{Generalization to more than two particles and to higher spin}
\label{gen}
As described in Theorem \ref{thee} maximally entangled states for two delocalized spin-$\frac{1}{2}$ particles only exists when there is three spatial modes if at most one particle an occupy any given mode. This is the most trivial example of a more general condition for existence of maximally entangled states in terms of a relation between the number of particles, their internal degrees of freedom, and the number of spatial modes.
A system of delocalized spin-$\frac{p}{2}$ particles admits maximally entangled states only for a ratio between the number of particles and the number of modes which is given by the particle type.

\begin{theorem}
\label{thee2}
In a system of $n$ spatial modes and $m$ delocalized spin-$\frac{p}{2}$ particles, where at most one particle occupy any spatial mode, entanglement measures and maximally entangled states exist if and only if $m/n=(p+1)/(p+2)$.
\end{theorem}
\begin{proof}
To show necessity of the relation $m/n=(p+1)/(p+2)$ we construct the generalization of the proof of Theorem \ref{thee}.
Consider $m$ delocalized spin-$\frac{p}{2}$ particles in a system of $n$ spatial modes and the diagonal $(p+2)\times(p+2)$ matrix $g$ acting on the local Hilbert space of some spatial mode

\begin{eqnarray}g=\left( \begin{array}{cccccc}
\!\! re^{i\phi} & \! 0 & 0 & \cdots & 0 & 0\\
\!\! 0 & re^{i\phi} & 0 & \cdots & 0 & 0\\
\!\! 0 & 0 &  \ddots &  & \vdots & \vdots \\
\!\! \vdots & \vdots & & \ddots &  0 & 0\\
\!\! 0 &  0   & \cdots & 0 & re^{i\phi} & 0\\
\!\! 0 & 0  & \cdots & 0 & 0 & r^{-(p+1)}e^{-(p+1)i\phi}\end{array} \!\right).\nonumber
\end{eqnarray}

The matrix is an element of $SL(p+2)$ and is a generalization of the matrix in the proof of Theorem \ref{thee}. The $p+1$ first columns of $g$ act on the internal degrees of freedom of the particle and the last column acts on $|0\rangle$. See e.g. Ref. \cite{pff} for a description of the generators of general $SU$ and $SL$ groups.

The operation $g^{\times{n}}$ belongs to the $G_{{\mbox{\tiny\itshape SLOCC}}}$ group for spin-$\frac{p}{2}$ particles.
The action of $g^{\times{n}}$ on a basis vector depends only on the number of occupied spatial modes. Therefore this action is a scaling factor $r^{(p+2)m-(p+1)n}e^{[(p+2)m-(p+1)n]i\phi}$ of the state vector. 

It follows that for $r\neq{1}$ it is impossible to construct polynomials in the state vector coefficients that are invariant under this action unless $m/n=(p+1)/(p+2)$. 

For the sufficiency of of the relation $m/n=(p+1)/(p+2)$ we make an explicit construction of maximally entangled states for arbitrary $n,m,p$ satisfying the relation.
For a given $p$, each particle has $p+1$ internal degrees of freedom. Consider the sequence $0,1,2,\dots,p-1,p,p+1$, where $0$ represents the absence of a particle and $1,\dots,p+1$ represents the internal degrees of freedom. Let $\sigma$ be the cyclic permutation defined by
$\sigma(0,1,2,\dots,p-1,p,p+1)=p+1,0,1,2,\dots,p-1,p$. Furthermore let $\sigma^{k}$ denote the composition of $k$ such cyclic permutations.

Each sequence corresponds to a Hilbert space vector which is a product of local states as given by the sequence.
Given this we can define a state $|\Psi_{\sigma}^{p+1}\rangle$ by
\begin{eqnarray}
|\Psi_{\sigma}^{p+1}\rangle=\frac{1}{\sqrt{p+2}}\sum_{i=1}^{p+2}|\sigma^{i}(0,1,2,\dots,p-1,p,p+1)\rangle,\nonumber
\end{eqnarray}
which is a maximally entangled state of $p+1$ particles with spin $\frac{p}{2}$ delocalized over $p+2$ spatial modes. To verify this is a straightforward exercise of calculating the one mode reduced density matrices and see that they are all maximally mixed.

To construct states with $r(p+1)$ particles with spin $\frac{p}{2}$ delocalized over $r(p+2)$ spatial modes we can compose a sequence $s(r,p+1)$ of $r$ copies of $0,1,2,\dots,p-1,p,p+1$ joined together $s(r,p+1)\equiv 0,1,2,\dots,p-1,p,p+1,0,1,2,\dots,p-1,p,p+1,0,1,2\dots$ and the state

\begin{eqnarray}
|\Psi_{\sigma}^{r,p+1}\rangle=\frac{1}{\sqrt{p+2}}\sum_{i=1}^{p+2}|\sigma^{i}[s(r,p+1)]\rangle\nonumber
\end{eqnarray}

In this way we can construct maximally entangled states for any $n,m,p$ satisfying $m/n=(p+1)/(p+2)$.
\end{proof}
From Theorem \ref{thee2} we can conclude that the particle species determines the required ratio of particles to the number of spatial modes for the existence of maximal entanglement.
 
Just as in the case of two particles delocalized over three spatial modes, maximally entangled states and entanglement monotones exist only if Bell-nonlocal correlations can be observed. Theorem \ref{t1} can be generalized to systems of arbitrary numbers of particles.

\begin{theorem}\label{t1g}
Let $|\psi\rangle$ be a state of $m$ delocalized spin-$\frac{p}{2}$ particles in $n$ modes where $m/n=(p+1)/(p+2)$. If $|\psi\rangle$ is Bell-local over a bipartition of the system into a single mode and $n-1$ modes, there exist no maximally entangled state to which $|\psi\rangle$ can be converted by SLOCC and no entanglement monotone takes a non-zero value for $|\psi\rangle$. Moreover, local special unitary operations can induce an arbitrary global phase shift of the state vector $|\psi\rangle$. Thus, if an entanglement monotone takes a non-zero value the state is Bell-nonlocal over every bi-partition of this type.
\end{theorem}
\begin{proof}

For a partitioning into a single mode and $n-1$ modes, a general state $|\psi\rangle$ can be written as
\begin{eqnarray}|\psi\rangle=\sum_{i=1}^{p+1}r_i|i\rangle|\theta_i\rangle+r_0|0\rangle|\phi\rangle,\end{eqnarray} 
where the $|i\rangle$ and $|0\rangle$ is a basis of the local Hilbert space of the single mode, the $|\theta_i\rangle$ are states of $m-1$ particles in $n-1$ modes and $|\phi\rangle$ is a state of $m$ particles in $n-1$ modes. Consider the projection of the state onto a subspace where one particle is localized in the single-mode part of the partitioning and $m-1$ in the remaining $n-1$ modes. All states in the subspace spanned by the $|\theta_i\rangle$ contains the same number of particles. Therefore, within this subspace the operations by a party with access to the $n-1$ partition are not restricted by the super-selection rule. 

Therefore nonlocal correlations on this subspace can be observed if and only if the projection of the state to this subspace is entangled. 
Thus the state is Bell-nonlocal if $|\theta_i\rangle\neq|\theta_j\rangle$ for some $i$ and $j$ and $r_ir_j\neq 0$.
If the state is Bell-local it thus follows that $r_ir_j=0$ or $|\theta_i\rangle=|\theta_j\rangle$ for each pair of $i,j$. Such a state is, up to local unitary operations, of the form $s_1|1\rangle|\varphi\rangle+s_2|0\rangle|\phi\rangle$ and is indistinguishable from a separable state $|s_1|^2|1\rangle\langle 1|\otimes|\varphi\rangle\langle\varphi|+|s_2|^2|0\rangle\langle 0|\otimes|\phi\rangle\langle\phi|$.



Next consider the diagonal $(p+2)\times(p+2)$ determinant one matrix
\begin{eqnarray}g(\alpha)=\left( \begin{array}{cccccc}
\!\! e^{-(p+1)\alpha} & \! 0 & 0 & \cdots & 0 & 0\\
\!\! 0 & e^{(p+3)\alpha} & 0 & \cdots & 0 & 0\\
\!\! 0 & 0 &  e^{\alpha} &  & \vdots & \vdots \\
\!\! \vdots & \vdots & & \ddots &  0 & 0\\
\!\! 0 &  0   & \cdots & 0 & e^{\alpha} & 0\\
\!\! 0 & 0  & \cdots & 0 & 0 & e^{-(p+1)\alpha}\end{array} \!\right),\nonumber
\end{eqnarray}
where $\alpha\in\mathbb{C}$. Here the first column corresponds to the action on the local state $|1\rangle$ and the last column corresponds to the action on $|0\rangle$.

For a state of the form $s_1|1\rangle|\varphi\rangle+s_2|0\rangle|\phi\rangle$ the SLOCC operation $g(\alpha)\times{I}\dots\times{I}$ changes the norm of the state vector by a factor of $|e^{-(p+1)\alpha}|$. By making $Re(\alpha)$ sufficiently large we can bring the norm arbitrarily close to zero. This implies that the set of states which are inter-convertible with $|\psi\rangle$ does not contain any maximally entangled states \cite{kempfness,verstraete2003}.
The form of the state $|\psi\rangle$ also implies that when $Re(\alpha)=0$ the SLOCC operation is special unitary and under this operation $g(\phi)\times{I}\dots\times{I}$ the state accumulates a phase-factor $e^{-(p+1)\alpha}$.
Thus, $e^{(p+1)\alpha}g(\alpha)\times{I}\dots\times{I}$ is an element of the stabilizer group. 

Vice versa, if there exist an entanglement monotone which takes a non-zero value for a state, the state cannot have the form $|\psi\rangle=s_1|1\rangle|\varphi\rangle+s_2|0\rangle|\phi\rangle$ for any bipartition into a single mode and $n-1$ modes, since no homogeneous polynomial is invariant under $g(\alpha)\times{I}\dots\times{I}$. Thus, the state must be Bell-nonlocal over every such bipartition.

\end{proof}

Theorem \ref{t1g} shows that just as in the case of two spin-$\frac{1}{2}$ particles and three spatial modes the the existence of SLOCC invariants and monotones implies that there exist bipartitions for which the system is Bell-nonlocal.

Theorem \ref{thee2} and Theorem \ref{t1g} gives some general properties of maximal entanglement, SLOCC invariants and monotones, and their connection to Bell-nonlocality in systems of delocalized particles under the constraint that at most one particle can occupy any given spatial mode. Such systems could be for example lattice systems described by a Hubbard model (Fermi- or Bose-) with a strong repulsive on-site interaction.
In this context maximally entangled states exist only for the particular lattice filling factor given as a function of the particle species by Theorem \ref{thee2}. Note that for a bosonic system such states would exist in the superfluid phase \cite{gubbels} due to the non-integer filling factor. 
Methods to spatially resolve occupation numbers in individual sites of optical lattices have been developed \cite{bakr,endres,maxwell,edge,omran}
as well as techniques to spatially resolve \cite{gross} and address \cite{weitenberg} spin degrees of freedom in individual lattice sites. These techniques, if combined, would enable the study of spatially resolved entanglement in spin and spatial degrees of freedom of the type discussed in this paper.


The theoretical methods described in Sect. \ref{sec2} for identifying the group of SLOCC and constructing the invariants of its action can be straightforwardly generalized. But the task of finding all SLOCC invariants is likely to become more demanding for larger numbers of particles. Likewise, constructing a parametrization of the entanglement types similar to that of Eq. \ref{onewee00} or describing the full set of maximally entangled states is likely to be more demanding.

\section{Conclusions}
\label{sec5}
We have considered a system of two indistinguishable spin-$\frac{1}{2}$ particles delocalized on three spatial modes where each spatial mode is occupied by at most one particle. 
We have shown that this is the only spatial configuration of two delocalized particles where maximally entangled states exist.

To describe the set of entangled states we identified the group of SLOCC and constructed two SLOCC invariants from which entanglement measures can be derived. These two invariants generate the full algebra of SLOCC invariants for this system.
We parametrized the set of entangled states and described how, for different parameter ranges, different configurations of Bell-local and Bell-nonlocal correlations can be observed for different partitions of the system and between different parties. In particular we showed that all maximally entangled states, and states for which an entanglement measure takes a nonzero value, are Bell-nonlocal for all bi-partitions of the system, but not necessarily for a tri-partition. Furthermore, the set of states which are Bell-local for a tri-partition of the system was described. These states cannot be described as a superposition of entangled states of two localized particles.

We parametrized and described the set of maximally entangled states. In general these states cannot be understood in terms of entangled pairs of spin-$\frac{1}{2}$ particles that have been delocalized. Two exceptional states with a high degree of symmetry were identified. One corresponding to a pair of maximally entangled particles that has been delocalized without further entangling spin and spatial degrees of freedom, and one maximally entangled state being Bell-local for three parties making local measurements. 

In addition to this we described how the delocalization gives rise to a non-trivial generic symmetry of the system. To this symmetry corresponds a $2\pi/3$ fractional topological phase.

Finally, a necessary and sufficient condition for the existence of maximally entangled states in systems of arbitrary numbers of delocalized particles with arbitrary spin, where at most one particle can occupy each mode, was derived. It was also shown in this general case that all states for which an entanglement measure takes a nonzero value, are Bell-nonlocal for any bi-partition into a single spatial mode and the remaining modes.

\subsection*{Acknowledgement}
The author acknowledges discussions with Micha\l{} Oszmaniec, Paul Skrzypczyk, and Hui Khoon Ng. The author also thank the anonymous referees for constructive comments.
Support from the Marie Curie COFUND action through the ICFOnest program, the John Templeton Foundation, the Spanish Project FOQUS (FIS2013-46768-P), Spanish MINECO (Severo Ochoa grant SEV-2015-0522), and the Generalitat de Catalunya (SGR 875) is acknowledged.

\section*{Appendix}
\label{appendix}
\subsection{Constructing the SLOCC invariants}

All invariants of the group $G_{123}$ can be constructed through Cayley's Omega Process. Cayley's Omega Process is designed to find the invariants of the special linear group and since $G_{123}$ is isomorphic to $SL(2)^{\times 3}$ we can use it to find the $G_{123}$ invariants.  We describe this process for the case of bilinear forms since this is sufficient for our purposes. For a more general description see e.g. \cite{olver}. 

The first step is to represent the state as a bilinear form $f$

\begin{eqnarray}
f=\sum_{i,j=\uparrow,\downarrow} m_{ij0}x_i y_j +m_{i0j}x_iz_j+m_{0ij}y_iz_j,
\end{eqnarray}
where $m_{ijk}$ is the probability amplitude of $|ijk\rangle$.
The invariants are obtained recursively from $f$ through iterating a type of operation called transvection defined by the following two steps.
First, apply a partial differential operator $\Omega_i$, where $i=x,y,z$. For example $\Omega_x$ is defined by

\begin{eqnarray}\Omega_{x}=\frac{\partial^2}{\partial{x'_{\uparrow}}\partial{x''_{\downarrow}}}-\frac{\partial^2}{\partial{x''_{\uparrow}}\partial{x'_{\downarrow}}}\end{eqnarray}
on the product of the two forms $A(x'_i,y'_j,z'_k)$ and $B(x''_i,y''_j,z''_k)$. 
This is followed by the substitution of $x$ for $x'$ and $x''$. If several operators $\Omega_i$ are applied, all of the involved variables are substituted. For example the application of $\Omega_x\Omega_z$ is followed by the substitution $x$ for $x'$ and $x''$ and $z$ for $z'$ and $z''$.

If the result of a transvection is a new form, further transvections can be performed. 
In the case of bilinear forms the result of a single transvection is either a scalar or another bilinear form. Therefore, if we represent a bilinear form $A$ as a set of matrices $\{A_{ij},A_{jk},A_{ki}\}$ and likewise for $B$,  the transvections can then be described by matrix operations. A transvection with a single $\Omega_i$ results in a bilinear form which can be represented in matrix form as $\{A_{ij}^T\sigma_y B_{ij},A_{ij}^T\sigma_yB^{T}_{ki},A_{ki}\sigma_yB_{ij},A_ki\sigma_yB^{T}_{ki}\}$. A transvection with two operators $\Omega_i$ and $\Omega_j$ corresponds to the matrix operation $Tr[(A_{ij}^T)\sigma_y(B_{ij})\sigma_y]$.

To use this technique for constructing invariants of $G_{123}$ we must represent the state of the system in a suitable form.
The twelve probability amplitudes of the state vector can be arranged as three $2\times 2$ matrices $M_{AB}$, $M_{AC}$, and $M_{BC}$.
The matrix $M_{AB}$ contains all the amplitudes of the events where C does not find a particle when performing a projective measurement. 
It is given by

\begin{eqnarray}M_{AB}=\left( \begin{array}{cc}
\!\!m_{\uparrow\uparrow{0}} & \!m_{\uparrow\downarrow{0}} \\
\!\!m_{\downarrow\uparrow{0}} &\!m_{\downarrow\downarrow{0}}  \end{array} \!\right).\end{eqnarray}
Under transformations $R_{A}$ and $R_{B}$ by A and B respectively $M_{AB}$  transforms as $M_{AB}\to{R}_{A}M_{AB}{R}_{B}^{T}$. The matrices $M_{BC}$ and $M_{AC}$ are defined analogously.

With this representation we can construct the four invariants

\begin{align}\label{i1}
 I_{AB}=&Tr(M_{AB}\sigma_{y}M_{AB}^{T}\sigma_{y})=m_{\uparrow\uparrow{0}}m_{\downarrow\downarrow{0}}-m_{\uparrow\downarrow{0}}m_{\downarrow\uparrow{0}},\nonumber\\
  I_{BC}=&Tr(M_{BC}\sigma_{y}M_{BC}^{T}\sigma_{y})=m_{0\uparrow\uparrow}m_{0\downarrow\downarrow}-m_{0\uparrow\downarrow}m_{0\downarrow\uparrow},\nonumber\\ I_{AC}=&Tr(M_{AC}\sigma_{y}M_{CA}^{T}\sigma_{y})=m_{\uparrow{0}\uparrow}m_{\downarrow{0}\downarrow}-m_{\uparrow{0}\downarrow}m_{\downarrow{0}\uparrow},\nonumber\\
\end{align}
and 

\begin{align}
I_{2}=&iTr(M_{AB}\sigma_{y}M_{BC}\sigma_{y}M_{AC}^T\sigma_y)
\nonumber\\=&m_{\uparrow\uparrow{0}}(m_{\downarrow{0}\uparrow}m_{{0}\downarrow\downarrow}-m_{{0}\downarrow\uparrow}m_{\downarrow{0}\downarrow})\nonumber\\+&m_{\uparrow\downarrow{0}}(m_{0\uparrow\uparrow}m_{\downarrow{0}\downarrow}-m_{\downarrow{0}\uparrow}m_{{0}\uparrow\downarrow})\nonumber\\+&m_{\downarrow\uparrow{0}}(m_{{0}\downarrow\uparrow}m_{\uparrow{0}\downarrow}-m_{\uparrow{0}\uparrow}m_{{0}\downarrow\downarrow})\nonumber\\+&m_{\downarrow\downarrow{0}}(m_{\uparrow{0}\uparrow}m_{{0}\uparrow\downarrow}-m_{0\uparrow\uparrow}m_{\uparrow{0}\downarrow}).
\end{align}

The polynomial $I_{AB}$ describe the bipartite entanglement between A and B when the two particles are localized with A and B respectively, and is the well known concurrence polynomial \cite{wootters}. Note that $Tr(M_{AB}\sigma_{y}M_{AB}^{T}\sigma_{y})=2 det(M_{AB})$. The polynomials $I_{BC}$ and $I_{AC}$ are the defined in the same way but for B and C, and A and C respectively.

These four invariants can be used to construct all other invariants of $G_{123}$. To see this we make the following observations. First, note that in any result of a transvection $M_{AB}$ must be multiplied by either by $\sigma_{y}M_{AB}^T$ or $\sigma_{y}M_{BC}$. No other possibilities exist. The same applies to $M_{BC}$ and $M_{AC}$. This means that for any invariant $Tr[X]$, the matrix $X$ can be decomposed into products of $M_{AB}\sigma_{y}M_{AB}^{T}\sigma_{y}$, $M_{BC}\sigma_{y}M_{BC}^{T}\sigma_{y}$, $M_{AC}\sigma_{y}M_{CA}^{T}\sigma_{y}$, and $M_{AB}\sigma_{y}M_{BC}\sigma_{y}M_{AC}^T\sigma_y$.

Next, note that $M_{AB}\sigma_{y}M_{AB}^{T}\sigma_{y}=det(M_{AB})I$, i.e., it is proportional to the identity matrix.  Therefore, a trace of a product of this matrix with another matrix $X$ can be factored as

\begin{eqnarray}Tr(M_{AB}\sigma_{y}M_{AB}^{T}\sigma_{y}X)
=Tr(M_{AB}\sigma_{y}M_{AB}^{T}\sigma_{y})Tr(X).\nonumber\\
\end{eqnarray}
The same applies to $M_{BC}\sigma_{y}M_{BC}^{T}\sigma_{y}$ and $M_{AC}\sigma_{y}M_{AC}^{T}\sigma_{y}$
The only invariants which cannot be algebraically decomposed into $I_{AB},I_{BC},I_{AC},I_2$ in this way are of the form

\begin{eqnarray}
Tr[(M_{AB}\sigma_{y}M_{BC}\sigma_{y}M_{AC}^T\sigma_y)^n],
\end{eqnarray}
for some $n$.
However, these invariants are decomposable into $I_{AB},I_{BC},I_{AC},I_2$ in a different way. To see this let us denote

\begin{eqnarray}
M_{AB}\sigma_{y}M_{BC}\sigma_{y}M_{AC}^T\sigma_y=\left( \begin{array}{cc}
\!\!F_1 & \!G_1 \\
\!\!K_1 &\!L_1  \end{array} \!\right).
\end{eqnarray}
Then $I_2=i(F_1+L_1)$. Furthermore, we note that $G_1K_1=I_{AB}I_{BC}I_{AC}+F_1L_1=I_{AB}I_{BC}I_{AC}-F_1(F_1+iI_2)$. Using these relations it can be shown by induction that
for any $n$ the invariant $Tr[(M_{AB}\sigma_{y}M_{BC}\sigma_{y}M_{AC}^T\sigma_y)^n]$ is a polynomial in $I_{AB},I_{BC},I_{AC}$ and $I_2$. 

As the first step consider 
\begin{eqnarray}
\left( \begin{array}{cc}
\!\!F_n & \!G_n \\
\!\!K_n &\!L_n  \end{array} \!\right)\equiv\left( \begin{array}{cc}
\!\!F_1 & \!G_1 \\
\!\!K_1 &\!L_1  \end{array} \!\right)^n.
\end{eqnarray}
Assume that $F_n$ and $L_n$, as well as $G_nK_1$ and $K_nG_1$ are polynomials in $I_{AB},I_{BC},I_{AC},I_2$ and $F_1$. Then for
\begin{align}
&\left( \begin{array}{cc}
\!\!F_{(n+1)} & \!G_{(n+1)} \\
\!\!K_{(n+1)} &\!L_{(n+1)}  \end{array} \!\right)=\left( \begin{array}{cc}
\!\!F_nF_1+G_nK_1 & \!F_nG_1+G_nL_1 \\
\!\!K_nF_1+L_nK_1 &\!K_nG_1+L_nL_1  \end{array} \!\right)\nonumber\\
&=\left( \begin{array}{cc}
\!\!F_nF_1+G_nK_1 & \!F_nG_1-G_n(F_1+iI_2) \\
\!\!K_nF_1+L_nK_1 &\!K_nG_1-L_n(F_1+iI_2)  \end{array} \!\right),
\end{align}
it is clear that the diagonal elements $F_{(n+1)}=F_nF_1+G_nK_1$ and $L_{(n+1)}=K_nG_1-L_n(F_1+iI_2)$ are polynomials in $I_{AB},I_{BC},I_{AC},I_2$ and $F_1$. Furthermore,

\begin{align} &G_{(n+1)}K_1=[F_nG_1-G_n(F_1+iI_2)]K_1=\nonumber\\
&=F_n[I_{AB}I_{BC}I_{AC}-F_1(F_1+iI_2)]-G_nK_1(F_1+iI_2), 
\end{align}
as well as
\begin{align}&K_{(n+1)}G_1=(K_nF_1+L_nK_1)G_1=\nonumber\\
&=K_nG_1F_1+L_n[I_{AB}I_{BC}I_{AC}-F_1(F_1+iI_2)],
\end{align} 
are then also polynomials in  $I_{AB},I_{BC},I_{AC},I_2$ and $F_1$.

Thus, if $F_n$ and $L_n$, as well as $G_nK_1$ and $K_nG_1$ are a polynomials in $I_{AB},I_{BC},I_{AC},I_2$ and $F_1$ for some $n$ it is true for $n+1$.
Next note that it holds true for $n=1$. By induction it therefore holds for any $n$.
It follows that  $Tr[(M_{AB}\sigma_{y}M_{BC}\sigma_{y}M_{AC}^T\sigma_y)^n]$ is a polynomial in $I_{AB},I_{BC},I_{AC},I_2$ and $F_1$ for any $n$. However, since this trace is a $G_{123}$ invariant by definition and $F_1$ is not an invariant, $F_1$ cannot be a factor of any term in such a polynomial. Thus, $Tr[(M_{AB}\sigma_{y}M_{BC}\sigma_{y}M_{AC}^T\sigma_y)^n]$ is a polynomial in $I_{AB},I_{BC},I_{AC},$ and $I_2$ for any $n$.

We can conclude that $I_{AB},I_{BC},I_{AC},$ and $I_2$ generate all the $G_{123}$ invariants.
The invariant $I_2$ is invariant also under $G_8$ but $I_{AB},I_{BC},I_{AC}$ are not. The product $I_1\equiv I_{AB}I_{BC}I_{AC}$ is the simplest product of these three polynomials that is invariant under $G_8$. Together $I_1$ and $I_2$ generate all the $G_8$ invariants contained in the algebra of $G_{123}$ invariants. These $G_{123}$ invariants which are also invariant under $G_8$ constitute the full set of $G_{{\mbox{\tiny\itshape SLOCC}}}$ invariants since $G_{123}$ and $G_8$ commute. Thus we can conclude that the invariants of $G_{{\mbox{\tiny\itshape SLOCC}}}$ are $I_1$ and $I_2$.

\subsubsection{Constructing unitary invariants}

Using similar methods one can construct unitary invariants by substituting some of the transvections, represented by the matrix operation $A_{ij}\sigma_y B^T_{kj}$, by $A_{ij}B^{\dagger}_{kj}$. This works since unitary matrices satisfy $UU^{\dagger}=1$ and thus $A_{ij}UU^{\dagger}B^{\dagger}_{kj}=A_{ij}B^{\dagger}_{kj}$. The invariants under he  action of the unitary subgroup of $G_{123}$ can therefore be found this way. The three invariants used in Section \ref{kk} are constructed as

\begin{align}
I_{3}&=Tr[(M_{AB}^T\sigma_{y}M_{AC})(M_{AB}^T\sigma_{y}M_{AC})^{\dagger}]\nonumber\\
=&|m_{\downarrow\uparrow{0}}m_{\uparrow{0}\uparrow}-m_{\uparrow\uparrow{0}}m_{\downarrow{0}\uparrow}|^2+|m_{\downarrow\uparrow{0}}m_{\uparrow{0}\downarrow}-m_{\uparrow\uparrow{0}}m_{\downarrow{0}\downarrow}|^2\nonumber\\+&|m_{\downarrow\downarrow{0}}m_{\uparrow{0}\uparrow}-m_{\uparrow\downarrow{0}}m_{\downarrow{0}\uparrow}|^2+|m_{\downarrow\downarrow{0}}m_{\uparrow{0}\downarrow}-m_{\uparrow\downarrow{0}}m_{\downarrow{0}\downarrow}|^2,\nonumber\\
I_{4}&=Tr[(M_{BC}\sigma_{y}M_{AC}^T)(M_{BC}\sigma_{y}M_{AC}^T)^{\dagger}]\nonumber\\
=&|m_{{0}\uparrow\downarrow}m_{\uparrow{0}\uparrow}-m_{0\uparrow\uparrow}m_{\uparrow{0}\downarrow}|^2+|m_{{0}\uparrow\downarrow}m_{\downarrow{0}\uparrow}-m_{0\uparrow\uparrow}m_{\downarrow{0}\downarrow}|^2,\nonumber\\+&|m_{{0}\downarrow\downarrow}m_{\uparrow{0}\uparrow}-m_{{0}\downarrow\uparrow}m_{\uparrow{0}\downarrow}|^2+|m_{{0}\downarrow\downarrow}m_{\downarrow{0}\uparrow}-m_{{0}\downarrow\uparrow}m_{\downarrow{0}\downarrow}|^2\nonumber\\
I_{5}&=Tr[(M_{AB}\sigma_{y}M_{BC})(M_{AB}\sigma_{y}M_{BC})^{\dagger}]\nonumber\\
=&|m_{\uparrow\downarrow{0}}m_{0\uparrow\uparrow}-m_{\uparrow\uparrow{0}}m_{{0}\downarrow\uparrow}|^2+|m_{\uparrow\downarrow{0}}m_{{0}\uparrow\downarrow}-m_{\uparrow\uparrow{0}}m_{{0}\downarrow\downarrow}|^2\nonumber\\+&|m_{\downarrow\downarrow{0}}m_{0\uparrow\uparrow}-m_{\downarrow\uparrow{0}}m_{{0}\downarrow\uparrow}|^2+|m_{\downarrow\downarrow{0}}m_{{0}\uparrow\downarrow}-m_{\downarrow\uparrow{0}}m_{{0}\downarrow\downarrow}|^2.\nonumber\\
\end{align}

It can easily be verified that these are also invariants of the unitary subgroup of $G_8$.

\subsection{SLOCC invariants when one of the particles is localized}
If one particle has a well defined position, i.e., is localized with a given party the group of local unitary operations and SLOCC operations is different and this has implications for the construction of SLOCC invariants. For example, let a particle be localized with party A. Then $\lambda_8$ is not a generator of the local action of the unitary operations or SLOCC on the mode of party A.

In this scenario we can construct the $G_{123}$ invariants as before, since this subgroup is still part of the SLOCC operations. However, since one particle is always with A the invariants $I_2$ and $I_{BC}$ are trivially zero. The $G_{123}$ invariants are thus generated by $I_{AB}$ and $I_{BC}$ alone. The next step is to construct polynomials that are invariant under the action generated by $\lambda_8$ on the modes of party $B$ and $C$ from $I_{AB}$ and $I_{BC}$. This is not possible since these two polynomials are not invariant under the operation $I\times e^{\alpha\lambda_8}\times e^{\alpha\lambda_8}$ and neither is any product of them. Thus, no SLOCC invariants exist in this scenario, and therefore no entanglement monotones \cite{vidal} or maximally entangled states \cite{verstraete2003,kempfness}.

\end{document}